\documentclass[letterpaper, 10 pt, conference]{ieeeconf}  

\IEEEoverridecommandlockouts                              

\usepackage{graphicx} 
\usepackage{epstopdf}
\usepackage{amsmath,amssymb}     
\newtheorem{thm}{Theorem}[section]
\newtheorem{lemma}[thm]{Lemma}
\newtheorem{cor}[thm]{Corollary}

\newtheorem{defin}[thm]{Definition}

\usepackage[T1]{fontenc}
\usepackage[latin1]{inputenc}
\usepackage{cite}

\newcommand{\R}{\mathbb{R}} 
\def\fif{{\rm fi}}

\title{\LARGE \bf
Explicit model predictive control accuracy analysis
}

\author{Andrew Knyazev$^{1,2}$, Peizhen Zhu$^{1}$ and Stefano Di Cairano$^{1,3}$
\thanks{$^{1}$Mitsubishi Electric Research Laboratories; 201 Broadway
Cambridge, MA 02139, USA}%
\thanks{$^{2}${\tt\small knyazev@merl.com}}	
	\thanks{$^{3}${\tt\small dicairano@merl.com}}
	\thanks{$^{}$Accepted to IEEE CDC 2015 conference}%
}

\begin{document}

\maketitle
\thispagestyle{empty}
\pagestyle{empty}

\begin{abstract}
Model Predictive Control (MPC) can efficiently control constrained systems in real-time applications.  
MPC feedback law for a linear system with linear inequality constraints can be explicitly computed off-line, which results in an off-line partition of the state space into non-overlapped convex regions, with affine control laws associated to each region of the partition.
An actual implementation of this explicit MPC in low cost micro-controllers requires the data to be 
``quantized'', i.e. represented with a small number of memory bits. 
An aggressive quantization decreases the number of bits and the controller manufacturing costs, and may increase 
the speed of the controller,  but reduces accuracy of the control input computation. 
We derive upper bounds for the absolute error in the control depending on the number of quantization bits and system parameters. 
The bounds can be used to determine how 
many quantization bits are needed in order to guarantee a specific level  of accuracy in the control input.
\end{abstract}
\section{INTRODUCTION}
Model predictive control (MPC) \cite{rawlings2009model} is an efficient method  for control design 
of 
multivariable constrained systems in chemical and process control, automotive, aerospace, and 
factory automation~\cite{QB03,di2012industry,hrovat2012}. MPC solves a 
constrained optimal control problem in real time (on-line).

Explicit MPC (EMPC)~\cite{bem2002,Bemporad2002} may reduce on-line computational costs and code 
complexity by pre-computing the MPC feedback law as a state feedback, thus making it viable for 
fast applications with limited computational 
capabilities~\cite{DYBKH11,DDKH12,DTBB12,di2012model}. 
In particular, for linear 
systems subject to linear constraints and cost function based on $1$-norm, $\infty$-norm, or 
squared $2$-norm, the EMPC results in a polyhedral piecewise affine (PWA) feedback law. 
Thus, during the on-line execution, the EMPC controller first identifies which polyhedral region 
contains the current state, and then computes the control action by evaluating the corresponding 
affine control law.  
The identification of the polyhedral region is referred to as the point location problem 
\cite{Goodman97},  which can be solved by sequential search 
and binary search tree see, e.g., \cite{Tonde03,bayat2011}. Due to the exponential increase of the 
number of regions with respect to the number of constraints in the MPC problems,  
techniques for reducing complexity of the EMPC feedback law while maintaining its 
most important properties have been proposed, see, e.g.,~\cite{geyer04,kvasnica11,kvasnica12} 
and references therein. 

In practice, the data of EMPC have to be typically stored in a micro-controller hardware  memory, so that every stored number is 
represented by a small fixed number of bits for every number in the data. 
In other words, the data cannot be stored exactly and hence a precision loss occurs. We call this 
reduction of 
precision ``quantization'' and the reduced precision data ``quantized'' data.  
The method for quantization can be as simple as rounding. 
Aggressive quantization has the 
advantage of decreasing memory requirements and increasing the speed of the control input 
evaluation, at the price of introducing inaccuracy in the computation of the control input. If the 
quantization precision is too small, the controller can fail to accurately determine the region for the 
current state of the controlled system, and thus, the control. For example, 
by quantizing the state measurement/estimate data the quantized state may  jump to a different region.

The effect of quantization has been investigated for implicit MPC for instance in~\cite{kerrigan12,longo14}.
We  investigate the resulting accuracy in the control input computation in  EMPC as a 
consequence of different quantization precisions, so that we can determine how many bits need to be 
used to  guarantee a desired level of accuracy in the control input. A~brief overview of our approach and results is in Section~\ref{sec:preliminaries}.
 
In Section~\ref{sec:analysis}, we provide a mathematical 
accuracy analysis, 
depending on mutual positions of the exact and the quantized system states.
When the quantization does not 
affect the system state region, so that the same feedback law applies to both the exact and the quantized system states, 
an error bound is easy to establish. 
A difficult case for analysis, leading to a much larger possible controller inaccuracy, is where the quantization makes the system state to jump over a region facet to a different region. In this case, bounding the  accuracy of the control requires taking into account not only quantization precision for the system state, but also quantization effects of the region facets and  of the feedback laws in different regions. We derive two kinds of upper bounds on the accuracy of the control input computation 
in Section~\ref{sec:analysis}. Bounds without knowledge of the quantized data, describing the worst case scenario, 
for that reason are called ``a priori.'' After a quantized implementation of the controller is determined, the a priori bounds are improved, using the already known off-line quantized data, in addition to the original data. 
The resulting tighter ``a posteriori bounds'' depend on the quantization precision of the current state. 
We show how our bounds can be improved by  exploiting a rescaling technique that makes the system state space evenly sized in all spacial directions. 

We validate the bounds numerically in Section~\ref{sec:test}. 

\section{Brief overview}\label{sec:preliminaries}
Throughout this paper, $\R$ denotes the set of real numbers,  $\|\cdot\|_1$, $\|\cdot\|_2$, and  $\|\cdot\|_\infty$ denote $1$-norm, $2$-norm, and  $\infty$-norm, respectively. $A^{\prime}$ denotes the transpose of $A.$ 

The control $u(x)=\left\{u_{1}(x),\ldots,u_{n_r}(x)\right\}$
is a continuous PWA function determined by the EMPC control law  
 \begin{equation}\label{eq:cl}
  u_i(x)=F_ix+G_i \quad \forall x\in P_i,\,  i=1,\ldots, n_r,
 \end{equation}
where $x\in \R^{n}$, the gains $F_i\in \R^{m\times n}$ and offsets $G_i\in \R^m,$ and $n_r$ denotes the number of regions
\[P_i=\mbox\{x\in \R^{n}| H_ix\leq K_i, \, H_i\in \R^{{n_c^i}\times n},  K_i\in  \R^{n_c^{i}}\},\]  
with
	$H_i=\left[{H_{i}^{1}}^{\prime},\ldots,{H_{i}^{{n_c}^{i}}}^{'}\right]^{\prime}$ and  $K_i=[ K_{i}^{1},\ldots, K_{i}^{{n_c}^{i}}]^{\prime}.$ 
	
The controller  determines on-line, i.e. in real time, which region contains the given state $x$. 
This is typically the most time consuming operation, for large $n$ and $n_r$, requiring computing numerous matrix-vector products  $H_ix$.  
If the state $x$ is in the region $P_i$, then the corresponding control law \eqref{eq:cl} is used to compute the control input $u(x)$. 
All the controller ``true'' data, determining the regions, and control law \eqref{eq:cl} are computed off-line, typically in the double precision computer arithmetic, and then are quantized and stored in a memory of the controller as quantized data.

The reduction of the precision of the data used by the controller decreases the number of bits stored by the controller and increases the speed of the control, but reduces the accuracy of the on-line execution of the controller. The target accuracy of the controller can be based on the accuracy of the sensor for sensing, or of the estimator for estimating the current state of the system. The state of the controlled system determined with quantized data deviates from the state of the controlled system determined with the true data within the limits depending on the control law, the data representing the system, and the precisions of the quantization. Thus, for any reduction of the precision of the quantized data, it is possible to determine off-line bounds for a maximal deviation of the state of the controlled system caused by that reduction. Using these bounds, different reductions in the precision, compared to the true precision, can be tested, and the maximal reduction of the quantization satisfying the accuracy requirement of the control can be selected.

It may also be advantageous to use different quantization precisions for various data. For example, the EMPC control law $u_i(x)=F_ix+G_i$ is evaluated only once on-line, for the determined index $i$, so the gains $F_i$ and offsets $G_i$ can be quantized in high precision and stored in a slow-access memory without noticeably affecting the controller speed. Moreover, different facets can benefit from using different numbers of bits in their quantized format to represent the same level of accuracy in the control, while speeding up the computationally challenging point location on-line search. 

However, in some bounds below, for simplicity of presentation we suppose that all data are quantized using the same scalar quantization function $f(z)$, assuming that $\hat{z}=f(z)=z+\Delta z$, where $|\Delta z|\leq \epsilon$ and $0\leq\epsilon\leq 1$, aiming at a fixed point quantization, rather than a floating point rounding. 
 From the true state $x$, we obtain the quantized state $\hat{x}=x+\Delta x$, and,  similarly, 
 $
 \hat{u}_i(\hat{x})=\hat{F}_i\hat{x}+\hat{G}_i,
 $
where $\hat{H}_i\hat{x}\leq \hat{K}_i$ for $i=1,\ldots, n_r$. We commonly use  the symbol ``hat'' to  denote the data after the quantization, and the symbol ``$\Delta$'' for the quantization error.

\section{Accuracy analysis}\label{sec:analysis}
In this section, we focus on  accuracy analysis of the control input computation. The control input depends on the location of the state vector, thus, to analyze the accuracy  of the control input computation we have to examine how the state changes and which region it falls into before and after the quantization.   
When it falls into the region with the same index, i.e.\ $\hat{x}$ is in the region $\hat{P}_i$ and $x$ is in the region $P_i$, the absolute accuracy of the control input is measured by the maximum absolute changes between $u_i(x)$ and $\hat{u}_i\left(\hat{x}\right)$, i.e., $\|\hat{u}_i\left(\hat{x}\right)-u_i(x) \|_\infty$, and can be easily bounded above by the precisions of the quantization.  

The control computation error is generally larger if the original state and the quantized state belong to different regions with large and different gains. We analyze the case where $\hat{x}$ is in the region $\hat{P}_i$ and $x$ is in the region $P_j$, with $i\neq j$, which means the state $x$ is in one of the regions, but after quantization the state $\hat{x}$ gets into another region, so the  accuracy of the control input is $\|\hat{u}_i\left(\hat{x}\right)-u_j(x) \|_\infty$, as illustrated in Fig.~\ref{fig:jumpOver}, where $n=m=1.$   
Fig.~\ref{fig:jumpOver} shows that the change in the slope (gain) when the state jumps over the hyperplane (a point in Fig.~\ref{fig:jumpOver} , since  $n=1$) can increase the control computation error. 

\begin{figure}[ht]
 \begin{center}$
 \begin{array}{cc}
 \includegraphics[width=1\linewidth]{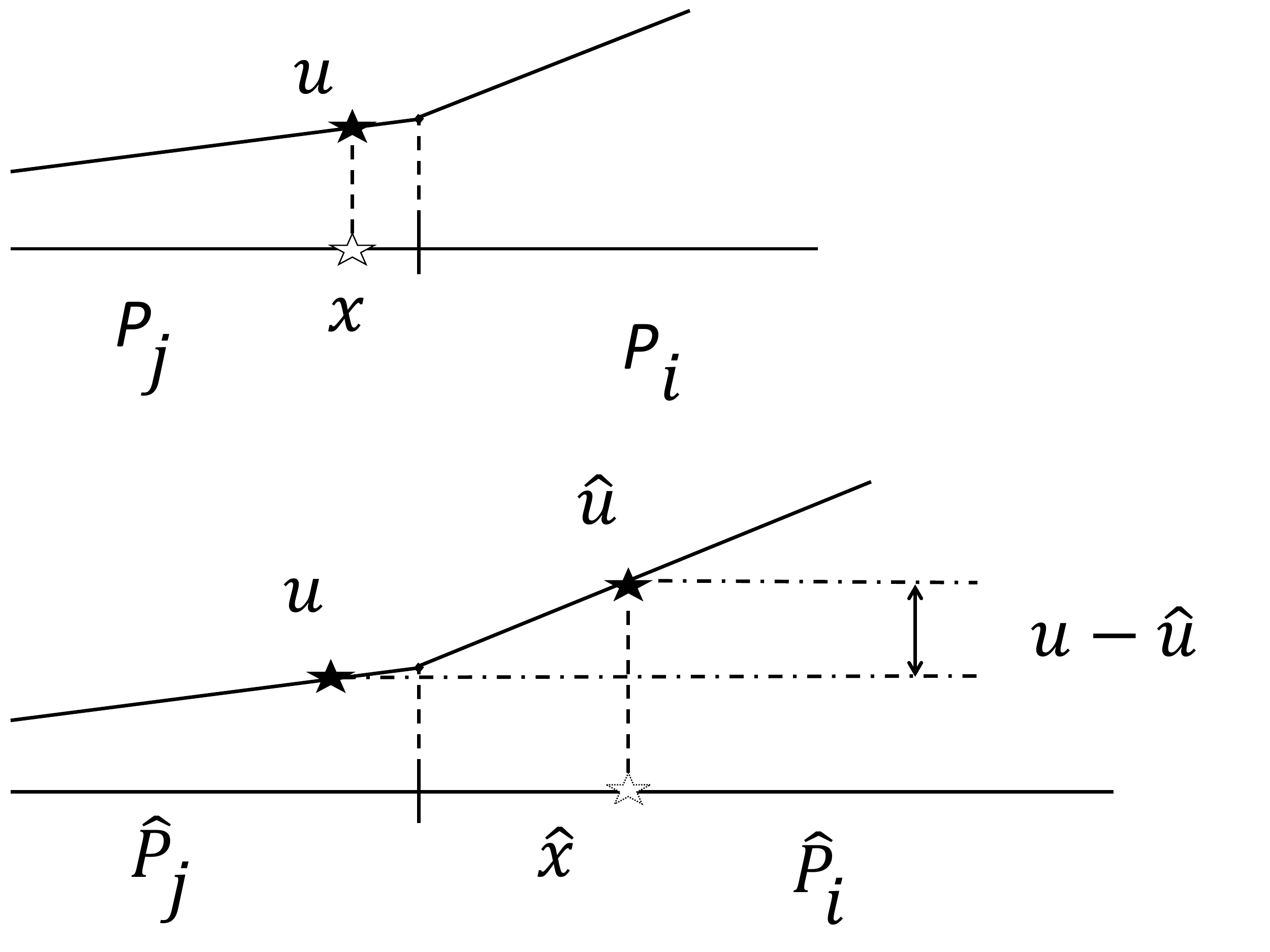}
 \end{array}$
 \end{center}												
		 \caption{ Location of the state within the regions. 
		 Top: the true data, where the true state $x$ is in the true region $P_i$ giving the true control $u$. Bottom: the quantized data, where the quantized state $\hat{x}$ jumps into the quantized region $\hat{P}_j$ having a different gain(slope) of the quantized control $\hat{u}$, resulting in a large error $u-\hat{u}$.}\label{fig:jumpOver}
 \end{figure}

However, the jump can only happen if a distance from the state vector to the facet separating the two regions is small, and the PWA control function is continuous, which allows us to bound the error even in the case of the jump. 
 
The sharp upper bound  of $\left\|\hat{u}(f(x))-u(x)\right\|_\infty$ is 
\[\max_{i,\, j}\max_{x\in P_j,\, f(x)\in\hat{P}_i}\|\hat{F}_i\ f(x)+\hat{G}_i-F_jx-G_j\|_\infty.\]
We derive an analytic bound for a pair $i$ and $j$ under simplifying assumptions, e.g.,\ that $P_i$ and $P_j$ share a facet. 		We start with analyzing the accuracy of the hyperplane representation, where the data are quantized.

\begin{lemma}\label{lem:1}
Let $hx\leq k$ be a half-space. Let $y=hx-k$ and $\hat{y}=(h+\Delta h)(x+ \Delta x)-(k+\Delta k)$  with $\|\Delta h\|_\infty\leq \epsilon$, $\|\Delta x\|_\infty\leq \epsilon$, and $|\Delta k|\leq \epsilon$  for some $\epsilon \geq 0$. We have%
\begin{equation}
|\hat{y}-y|\leq\delta := \epsilon (\|h\|_1+\|x\|_1+n\epsilon+1). \label{eqn:lem1eq}
\end{equation}
\end{lemma}
\begin{proof}
By direct calculation, we obtain
\begin{align*}
 |\hat{y}-y|&=|(h+\Delta h)(x+\Delta x)-(k+\Delta k)-(hx-k)|\\
&=|h\Delta x+\Delta h x+\Delta h \Delta x-\Delta k|\\
               &\leq |h\Delta x|+|\Delta h x|+|\Delta h \Delta x|+|\Delta k|\\
							&\leq \|h\|_1 \|\Delta x\|_\infty + \\
							&\qquad \|\Delta h\|_\infty\|x\|_1+\|\Delta h\|_\infty \|\Delta x\|_1+|\Delta k|\\
							&\leq \epsilon\|h\|_1+  \epsilon\|x\|_1+n\epsilon^2+\epsilon.
\end{align*}
\end{proof}
\begin{cor}\label{cor:dist}
Let the hyperplane $hx=k$ border two neighboring regions $P_i$ and $P_j$, such that $hx\leq k$ for $x\in P_j$ and $hx\geq k$ for $x\in P_i.$ 
Let the state $x\in P_j$, but after the quantization $\hat{x}\in\hat{P}_i$ with $i\neq j$. 
Then $-\delta< y\leq 0\leq \hat{y}< \delta$.	
\end{cor}
\begin{proof}By Lemma \ref{lem:1}, 
$|\hat{y}-y|\leq \delta$.
After quantization, $\hat{x}$ falls into the region $\hat{P}_i$ which means $\hat{x}$ is out of region $\hat{P}_j$. So,
$\delta+y>0$. Therefore, we have
$-\delta< y\leq 0.$
Similarly, we can obtain that  $0\leq \hat{y}< \delta$. 
\end{proof}

Corollary \ref{cor:dist} states  that, as a result of the quantization, the state can jumps to the other side of a hyperplane only if the distance from the state to the hyperplane is less than $\delta$, as illustrated in Fig.~\ref{fig:2dreg} for $n=2.$	 
\begin{figure}[ht]
 \begin{center}$
 \begin{array}{cc}
 \includegraphics[width=1\linewidth]{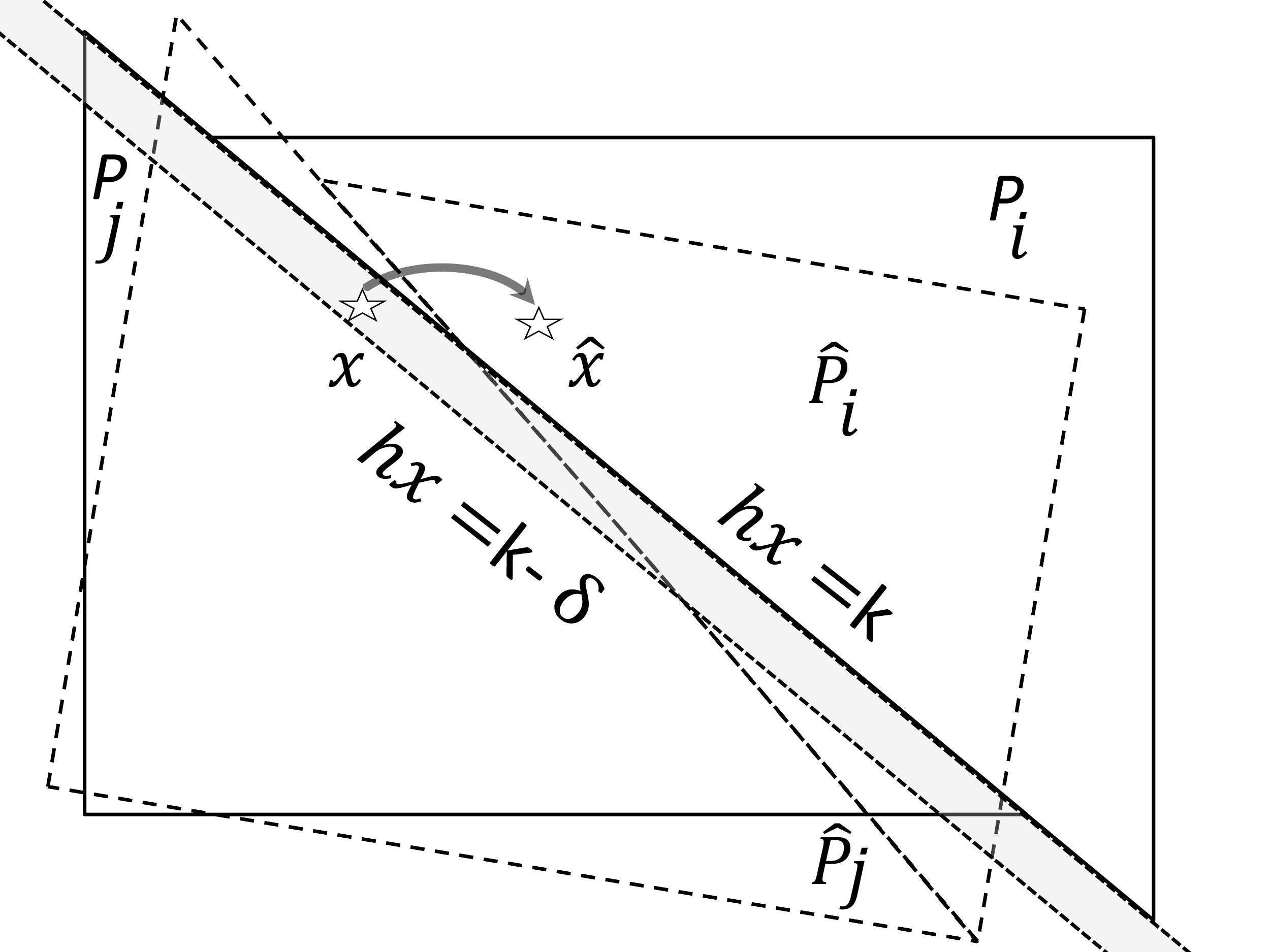}
 \end{array}$
 \end{center}												
		 \caption{ The state $x$ in $P_i$ after quantization turns into $\hat{x}$ in $\hat{P}_j$. }\label{fig:2dreg}
 \end{figure}

We are now ready to prove our main result. 

\begin{thm}\label{thm:thm1_update}
Let the hyperplane $hx=k$ border two neighboring regions $P_i$ and $P_j$, such that $hx\leq k$ if $x \in P_j$ and $hx \geq k$ if  $x \in P_i$.
Let the state $x$ be in the region $P_j$,
the orthogonal projection of $x$ on the hyperplane $hx=k$ be in a facet of $P_i$, and 
$\hat{x}=f(x)\in\hat{P}_i$. Then 
 \begin{align*}
			\|\hat{u}(\hat{x})-u(x)\|_\infty
												&\leq	\frac{\delta }{\|h\|_2^2}\|(F_i-F_j)h'\|_\infty+ \nonumber\\
												&\quad \|F_i \Delta x+\Delta F_ix +\Delta F_i\Delta x +\Delta G_i\|_\infty,
											      \end{align*}
														where $\delta=\epsilon (\|h\|_1+\|x\|_1+n\epsilon+1)$ and further 
 \begin{align}\label{eqn:aposteriori}
 							&\|F_i \Delta x+\Delta F_ix +\Delta F_i\Delta x +\Delta G_i\|_\infty \nonumber\\
												&\quad \quad\leq\|\Delta F_i\|_\infty \|x\|_\infty+\|\Delta G_i\|_\infty+\|\hat{F}_i\|_\infty\epsilon,					
							\end{align}	
							or, alternatively, 	
							 \begin{align}\label{eqn:apriori}
 							&\|F_i \Delta x+\Delta F_ix +\Delta F_i\Delta x +\Delta G_i\|_\infty \nonumber\\
												&\quad \quad\leq\epsilon(\|F_i\|_\infty+n\|x\|_\infty+n\epsilon+1).							
							\end{align}						
\end{thm}
\begin{proof} 
By the triangle inequality,  we have
\begin{align*}
			\|\hat{u}(\hat{x})-u(x)\|_\infty
			&=\|\hat{u}_i(\hat{x})-u_j(x)\|_\infty\\
			&=\|\hat{F}_i\hat{x}+\hat{G}_i-F_jx-G_j\|_\infty \\
														&\leq \|F_ix+G_i-F_jx-G_j\|_\infty+\\
														&\quad \|F_i \Delta x+\Delta F_ix +\Delta F_i\Delta x +\Delta G_i\|_\infty.
																									      \end{align*}	
We first bound above $\|F_ix+G_i-F_jx-G_j\|_\infty$, where $x$ satisfies	$-\delta< hx- k\leq 0$ by Corollary \ref{cor:dist}. 
Let $x_p$ denote the orthogonal projection of $x$ on the hyperplane $hx=k$, then 
$x_p=x+th',$
where $h'$ is the transpose of $h$ and $t$ is a scalar. We have 
$k=hx_p=hx+thh'=hx+t\|h\|^2_2.$		
Therefore, $|t|=|hx-k|/\|h\|_2^2.$ Since $u(x)$ is a linear continuous affine function, 
we have 
$
F_ix_p+G_i=F_jx_p+G_j$, which is equivalent to
$
F_i(x+th')+G_i=F_j(x+th')+G_j.
$
Hence,
\begin{align*}
\|F_ix+G_i-F_jx-G_j\|_\infty&=\|t(F_i-F_j)h'\|_\infty\\ 
                            &=|t|\|(F_i-F_j)h'\|_\infty\\ 
														&=\frac{|hx-k|}{\|h\|_2^2}\|(F_i-F_j)h'\|_\infty\\ 
														&\leq \frac{\delta}{\|h\|_2^2}\|(F_i-F_j)h'\|_\infty.
\end{align*}

The second term in the triangle inequality is  
\begin{align*}
					&\|F_i \Delta x+\Delta F_ix +\Delta F_i\Delta x +\Delta G_i\|_\infty	\\
						&\quad=	\|(\Delta F_ix+\Delta G_i)+(F_i \Delta x+\Delta F_i\Delta x)\|_\infty\\
						&\quad\leq\|\Delta F_i x+\Delta G_i\|_\infty+\|\hat{F}_i\|_\infty\|\Delta x\|_\infty\\
						&\quad\leq\|\Delta F_i\|_\infty \|x\|_\infty+\|\Delta G_i\|_\infty+\|\hat{F}_i\|_\infty\epsilon.
	   \end{align*}
 Alternatively, we bound the second term as follows
   \begin{align*}
		&\|F_i \Delta x+\Delta F_ix +\Delta F_i\Delta x +\Delta G_i\|_\infty \\
			&\quad\leq\|F_i \Delta x\|_\infty+\|	\Delta F_ix\|_\infty+	\|\Delta F_i\Delta x\|_\infty+\|\Delta G_i\|_\infty\\				
														&\quad\leq\|F_i\|_\infty\|\Delta x\|_\infty+\|\Delta F_i\|_\infty\|x\|_\infty  \\
														&\quad\quad+\|\Delta F_i\|_\infty \|\Delta x\|_\infty+ \|\Delta G_i\|_\infty \\
												&\quad\leq	\epsilon(\|F_i\|_\infty+n\|x\|_\infty+n\epsilon+1).
											      \end{align*}										      
		Combining the inequalities above completes the proof.
\end{proof}

	Bound \eqref{eqn:aposteriori} uses the data obtained before and after quantization, therefore we call it \emph{a posteriori}. 
In contrast, bound~\eqref{eqn:apriori}, called \emph{a priori}, is more pessimistic, being based on the true data only, without the quantized data explicitly appearing.

In bounds \eqref{eqn:aposteriori} and \eqref{eqn:apriori} and in the expression for $\delta$,  the norms of the state vector can be further bounded, e.g.,\ since $x\in P_j$, by the largest norm of any vector in the region $P_j$. If it is not known which region the state belongs to, we can use as a universal upper bound the largest norm of any vector in the state space, assuming that the latter is bounded. 

Theorem \ref{thm:thm1_update} implies that the control error is significantly reduced if the gains $F_i$ and offsets $G_i$ are quantized with high precision. Theorem \ref{thm:thm1_update} also suggests choosing different precisions for quantizing different facets  maintaining a uniform level of control accuracy over the state space, e.g., if $\|(F_i-F_j)h'\|_\infty$ is small, the hyperplane $hx=k$ can be quantized with low precision, without significantly affecting the control computation accuracy. 

\subsection{Rescaling}\label{sub:res}
The accuracy of $u(x)$ depends on how small $\epsilon$ and $\delta$ are. To make $\epsilon$ and $\delta$ small, one approach is to rescale the regions, $h$, and $k$, 
such that $$\frac{hD^{-1}Dx}{\max(\|hD^{-1}\|_1,|k|)}\leq \frac{k}{\max(\|hD^{-1}\|_1,|k|)},$$
where $D$ is a diagonal matrix with  $\|Dx\|_\infty\leq 1.$
 Let $Dx$,  $hD^{-1}/\max\left(\|hD^{-1}\|_1,|k|\right)$, and $k/\max\left(\|hD^{-1}\|_1,|k|\right)$ be our new $x$, $h$, and $k$, respectively. We have $\|x\|_\infty\leq 1$, $|hx|\leq 1$, and $|k|\leq 1$.  Our new $\delta$ is  less than $\epsilon(n+2+n\epsilon)$. The control law has the form 
 $u_i(x)=F_iD^{-1}x+G_i$,  where $x\in P_i,\,  i=1,\ldots, n_r.$
Consequently, 
 \begin{align*}
			\|\hat{u}(\hat{x})-u(x)\|_\infty
			&\leq\frac{\delta}{\|h\|_2^2}\|(F_i-F_j)D^{-1}h'\|_\infty+\\
			          & \quad \epsilon\|F_iD^{-1}\|_\infty+n\epsilon_1\|x\|_\infty+n\epsilon\epsilon_1+\epsilon_1,		
															\end{align*}		
where $\epsilon$ is taken, after rescaling, for the regions  and $\epsilon_1$ is taken, after rescaling, for the input control.
					
Rescaling of the state is evidently equivalent to changing the units of the components of the state vector. The bounds show that it may be beneficial to choose the units in such a way that the state space is balanced in size in all components. 
	
	\section{Test Results}\label{sec:test}
	
	In this section, we present tests for a double integrator 
\begin{eqnarray*}
x(k+1)&=&Ax(k) +Bu(k)\\
y(k)&=&Cx(k)+Du(k),
\end{eqnarray*}	
where 
\[
A=\left[\begin{array}{ll}
1&1\\
0&1\end{array}
\right], B=\left[\begin{array}{l}
0\\
1
\end{array}\right], C=\left[\begin{array}{ll}
1&0\\
0&1\end{array}
\right], D=\left[\begin{array}{l}
0\\
0
\end{array}\right].
\]
The state $x$ satisfies  the constraints $\left[\begin{array}{l} -15\\ -15 \end{array}\right]\leq x\leq \left[\begin{array}{l} 15\\ 15 \end{array}\right]$ and the input $u$ satisfies the constraints $-1\leq u\leq 1.$

		In our tests, we use a fixed point number format which  has a specific number of bits  reserved for the integer part and a specific number of bits reserved for the fractional part. We use the MATLAB function fi with $a$-bit total word length, $1$-bit for sign and $b$-bit fraction length, such that, e.g.,\ 
		\[\hat{x}=\fif(x,1,a,b)=x+\Delta x,\] 
		where $\|\Delta x\|_\infty\leq \epsilon=2^{-b}$ for a vector $x$. 
The quantization errors $\Delta H_i,$ $\Delta K_i$, $\Delta F_i$, and $\Delta G_i$ are known,  given the number $b$ of bits for the fractional part. Every component in $\Delta H_i,$ $\Delta K_i$, $\Delta F_i$, and $\Delta G_i$ is bounded by $2^{-b}.$		
 
We generate a consistent uniformly distributed random state $x$, and for every such a state compute the a priori and a posteriori bounds of absolute errors of $u(x)$ as described in Theorem \ref{thm:thm1_update}, as well as the actual absolute error of $u(x)$ obtained by a fixed point model of the controller. We~eliminate the states with very small  a posteriori bounds and actual errors (due to the fact that in some regions the gain is nearly zero), order the remaining states by sorting the a priori errors, and plot in Fig.~\ref{fig:2dexp} for two values of $a$ and $b$.
 \begin{figure}[ht]
 \begin{center}$
 \begin{array}{c}
 \includegraphics[width=1\linewidth]{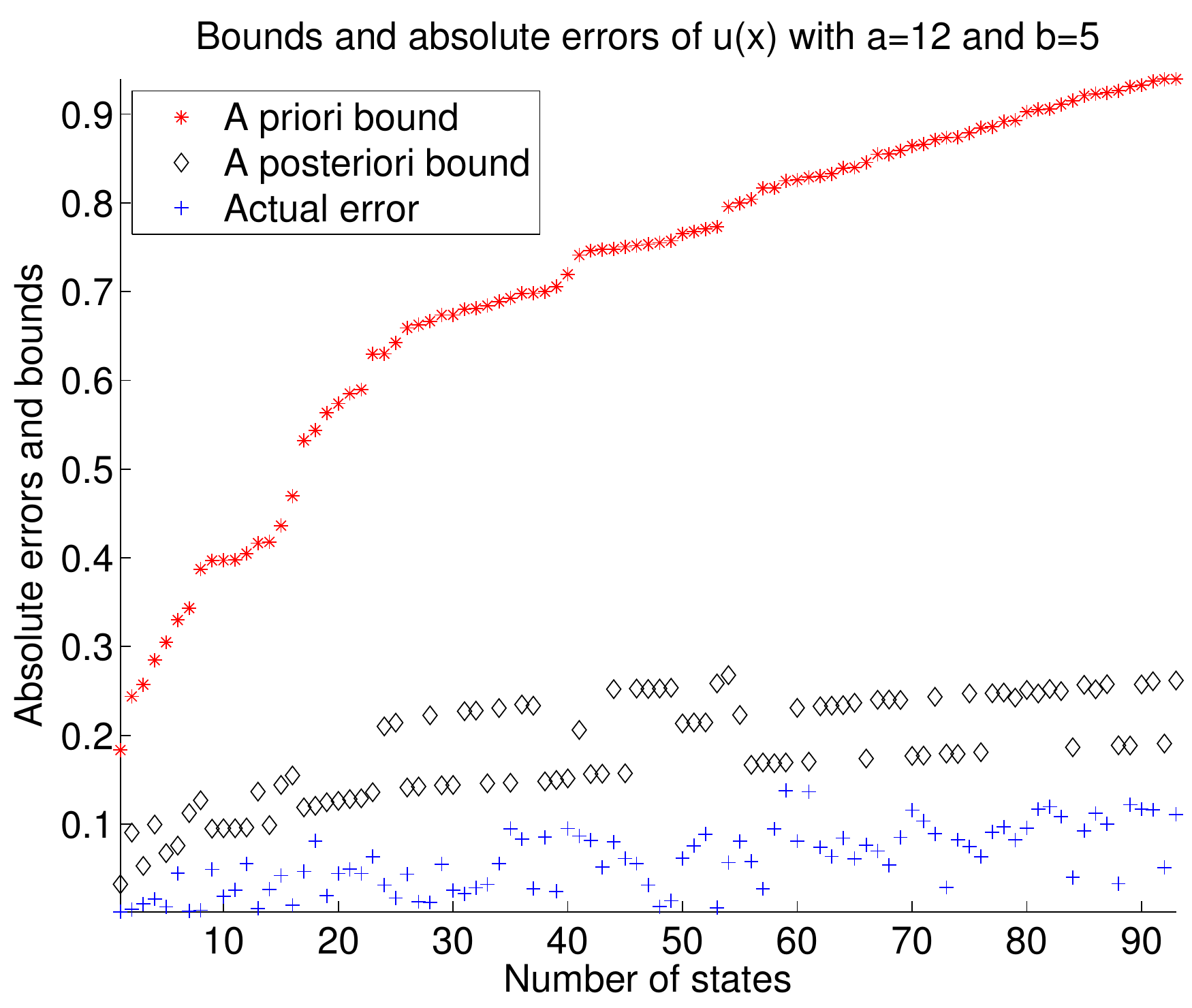} \\
 \includegraphics[width=1\linewidth]{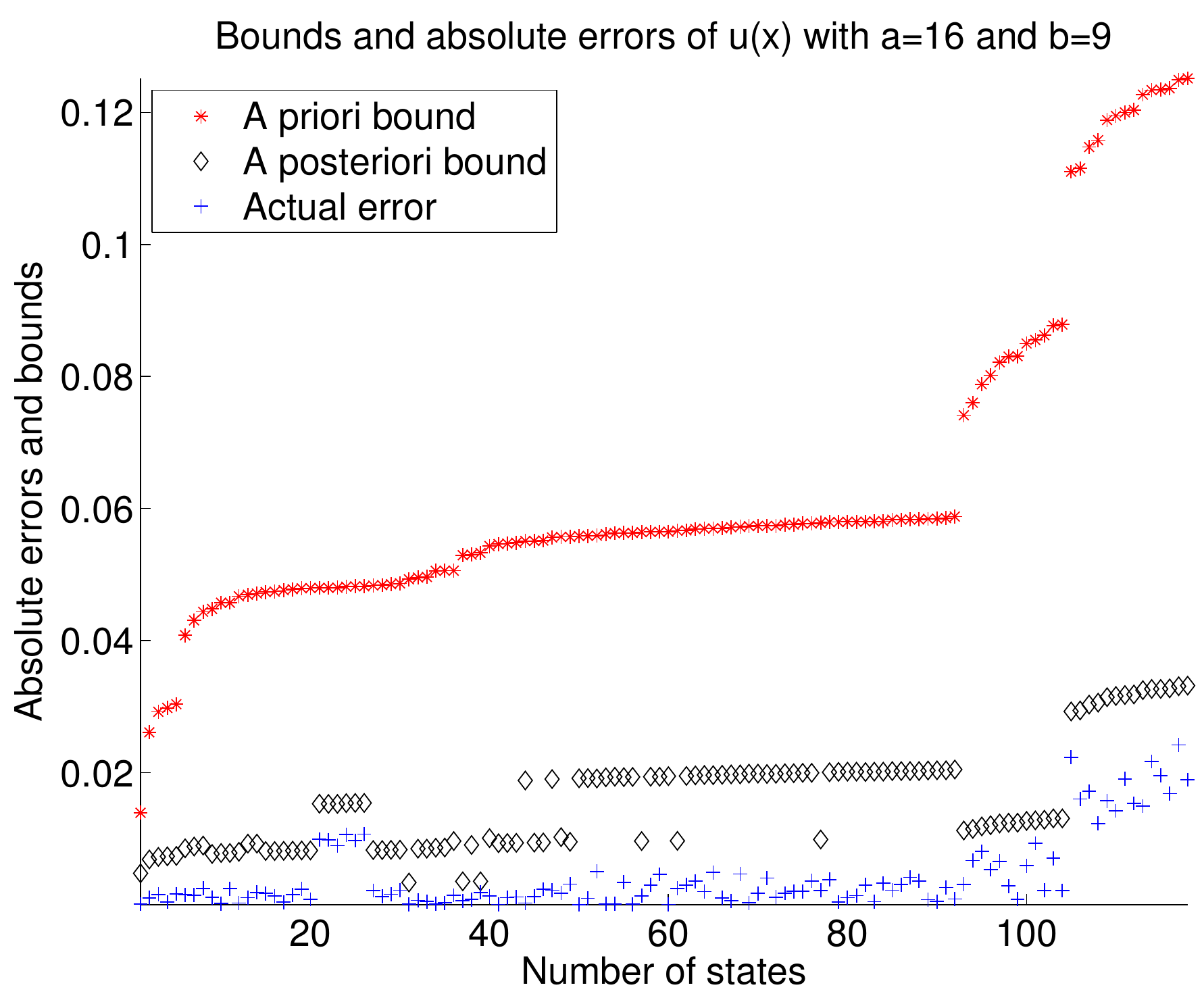}
 \end{array}$
 \end{center}
 \caption{Evaluating a-bit word total length and b-bit fraction length. Top:  a priori bounds, a posteriori bounds and  real computation absolute errors with $a=12$ and $b=5$.
 Bottom:  both bounds of absolute errors of $u(x)$ and  absolute errors in real computation with $a=16$ and $b=9$.   
				}\label{fig:2dexp}
 \end{figure} 

In Fig. \ref{fig:2dexp} top (bottom) panel for  $a=12$ ($a=16$) and $b=5$ ($b=9$),  the maximal difference between our a priori bounds and actual errors is about $0.9$ ($0.12$), the maximal difference between our a posteriori bounds and actual errors is about $0.26$ ($0.03$), and the maximal actual error is about $0.1$ ($0.02$), i.e., approximately $300\%$ ($1000\%$) compared to the quantization error $2^{-5}\approx 0.03$ ($2^{-9}\approx 0.002$).  

We observe in  Fig. \ref{fig:2dexp} bottom panel that the bounds and the actual errors form clusters. These clusters correspond to different hyperplanes, demonstrating that it may be advantageous to quantize hyperplanes using a plurality of precisions aiming at a uniform behavior of the bounds and the actual errors of the control over the state space. Next, we specifically test the states near the common hyperplanes between some pairs of neighbor regions to classify these hyperplanes in terms of their sensitivity to the quantization.

The results are summarized in Table \ref{tab:title}. For example, the bounds and the real errors are less than $10^{-4},$ if the state is in region $P_1$ or $P_9$ and after quantization the state jumps into region $\hat{P}_9$ or $\hat{P}_1.$ Checking the controller data, we find that the gains in the regions $P_1$ and $P_9$ are nearly zero, thus, naturally, the control does not change if the state is in $P_1$ and $P_9$, which is also well captured by our a posteriori bound, thus the hyperplane separating the regions $P_1$ and $P_9$ can be quantized with very low precision. We notice several other trivial pairs of the regions in Table \ref{tab:title}, identified by the small bounds and errors. 
	\begin {table}[ht]
	\begin{center}
	\begin{tabular}{ |c|c|c|}	
   \hline 
 Neighboring & Maximal  & Maximal  \\
      regions  & a posteriori bound  &  actual  error    \\   \hline
  1 and 9&  less than $10^{-4}$& less than $10^{-4}$\\ \hline
	2 and 8& less than  $10^{-4}$ &less than $10^{-4}$\\ \hline
	 3 and 5& 0.19& 0.15\\ \hline
		3 and 12 &0.25& 0.12\\	 \hline			
	4 and 7& less than $10^{-4}$& less than $10^{-4}$ \\ \hline
	5 and 6&less than $10^{-4}$& less than $10^{-4}$ \\ \hline
	6 and 9& less than $10^{-4}$& less than $10^{-4}$\\ \hline
	7 and 8 &less than $10^{-4}$& less than $10^{-4}$ \\ \hline
	7 and 13& 0.21& 0.12 \\ \hline
	10 and 13& 0.26& 0.13\\ \hline
   \hline
\end{tabular}
\end{center}\caption{A posteriori bounds and real errors for the  states near a single common hyperplane with $a=12$ and $b=5$}\label{tab:title} 
\end{table}
Let us check closely one nontrivial pair, e.g., $P_3$ and $P_5$. The maximal a posteriori bound is $0.19$ and the largest actual error is $0.15$,  if the state is in region $P_3$ or $P_5$ and after quantization the state jumps into region $\hat{P}_3$ or $\hat{P}_5.$  Fig.~\ref{fig:7dexp} plots the a posteriori bounds and the actual errors for this case. We observe that our a posteriori bounds always bound above the actual errors and are reasonably sharp. 
		\begin{figure}[ht]
 \begin{center}
 \includegraphics[width=1\linewidth]{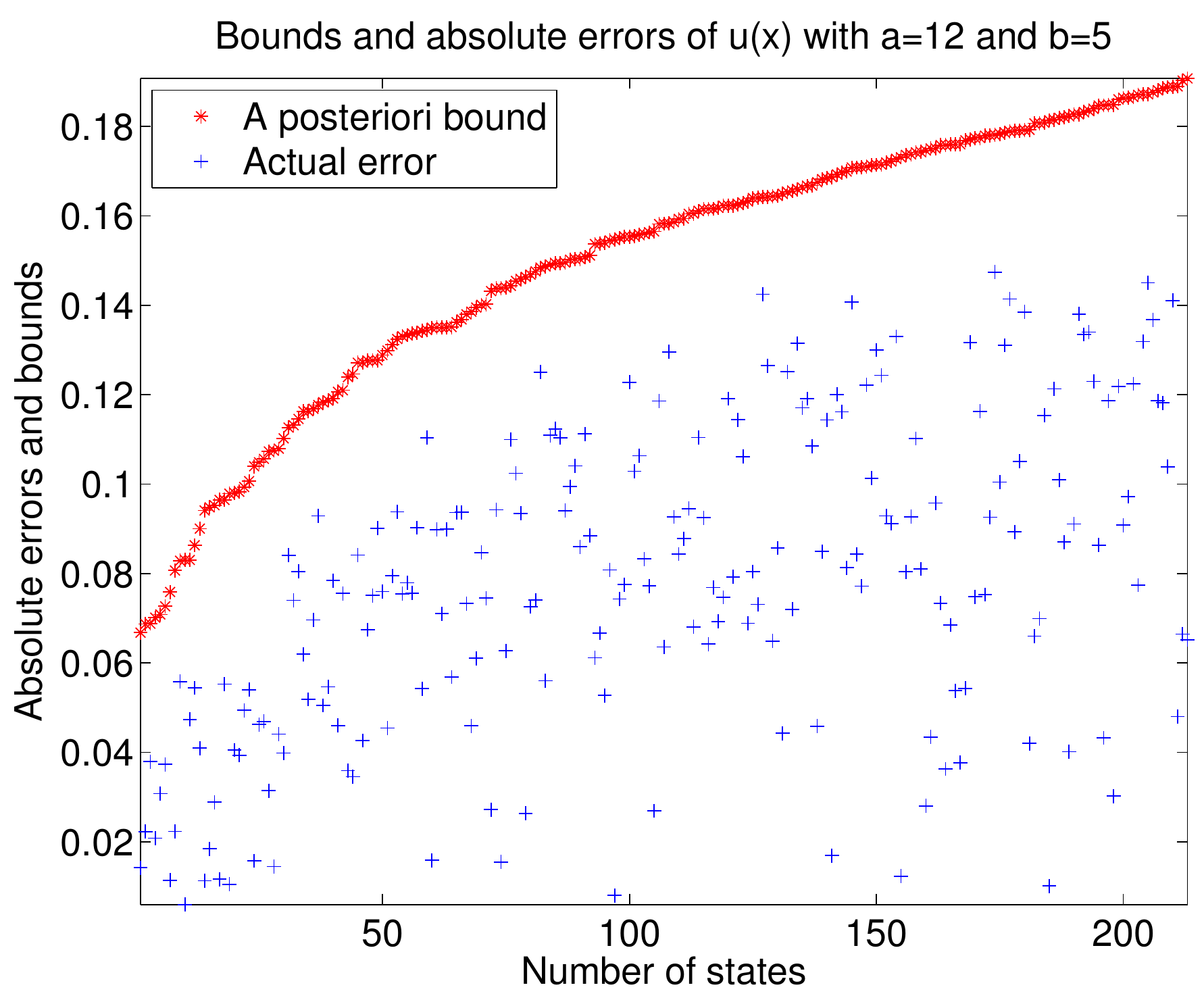} 
 \end{center}
 \caption{A posteriori bounds and real errors with $a=12$ and $b=5$ for the states near the common hyperplane between regions
3 and 5.}\label{fig:7dexp}
 \end{figure}

		\begin{figure}[ht]
 \begin{center}
 \includegraphics[width=1\linewidth]{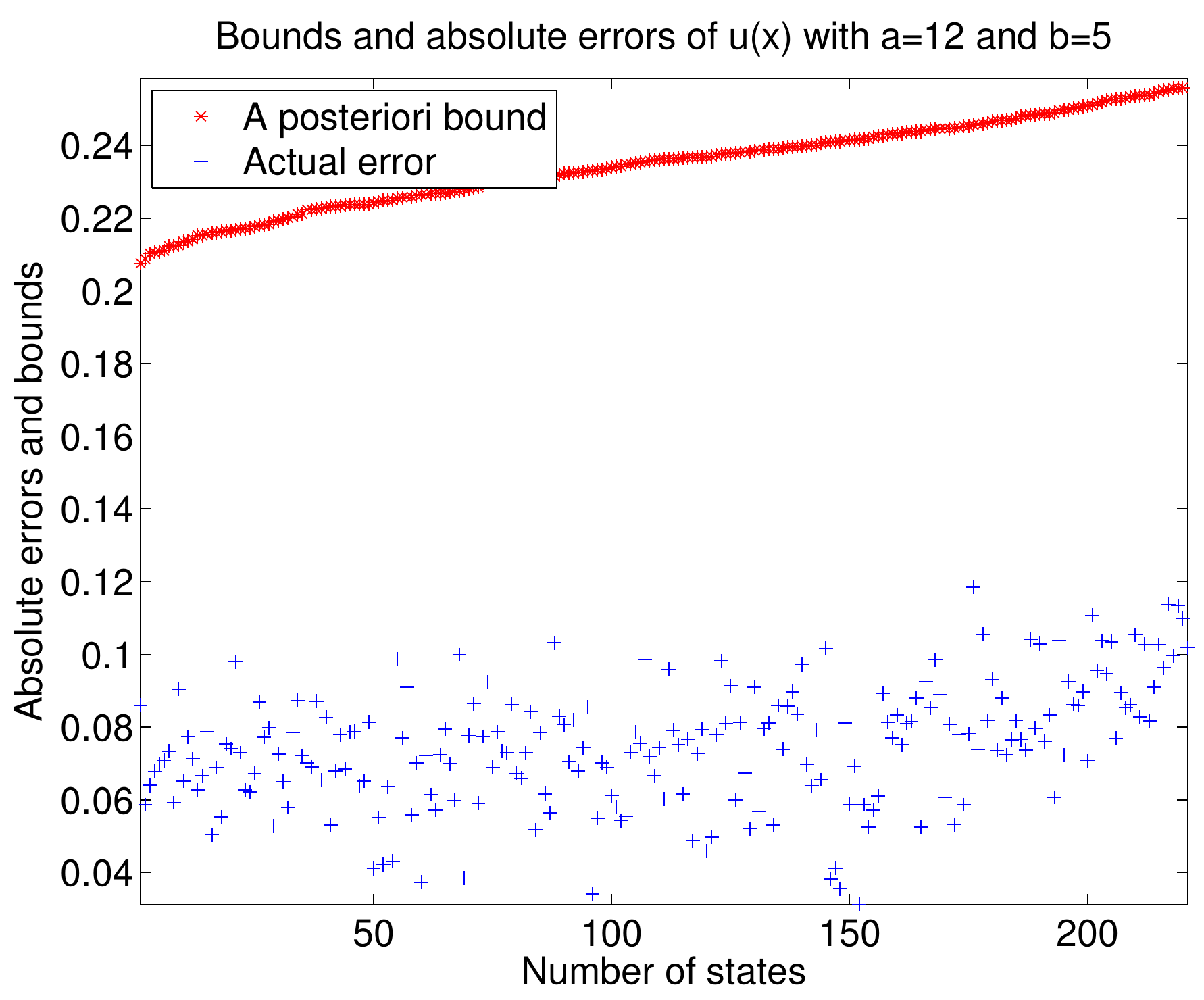} 
 \end{center}
 \caption{A posteriori bounds and real errors with $a=12$ and $b=5$ for the states near the common hyperplane between regions
3 and 12.}\label{fig:4}
 \end{figure}
		\begin{figure}[ht]
 \begin{center}
 \includegraphics[width=1\linewidth]{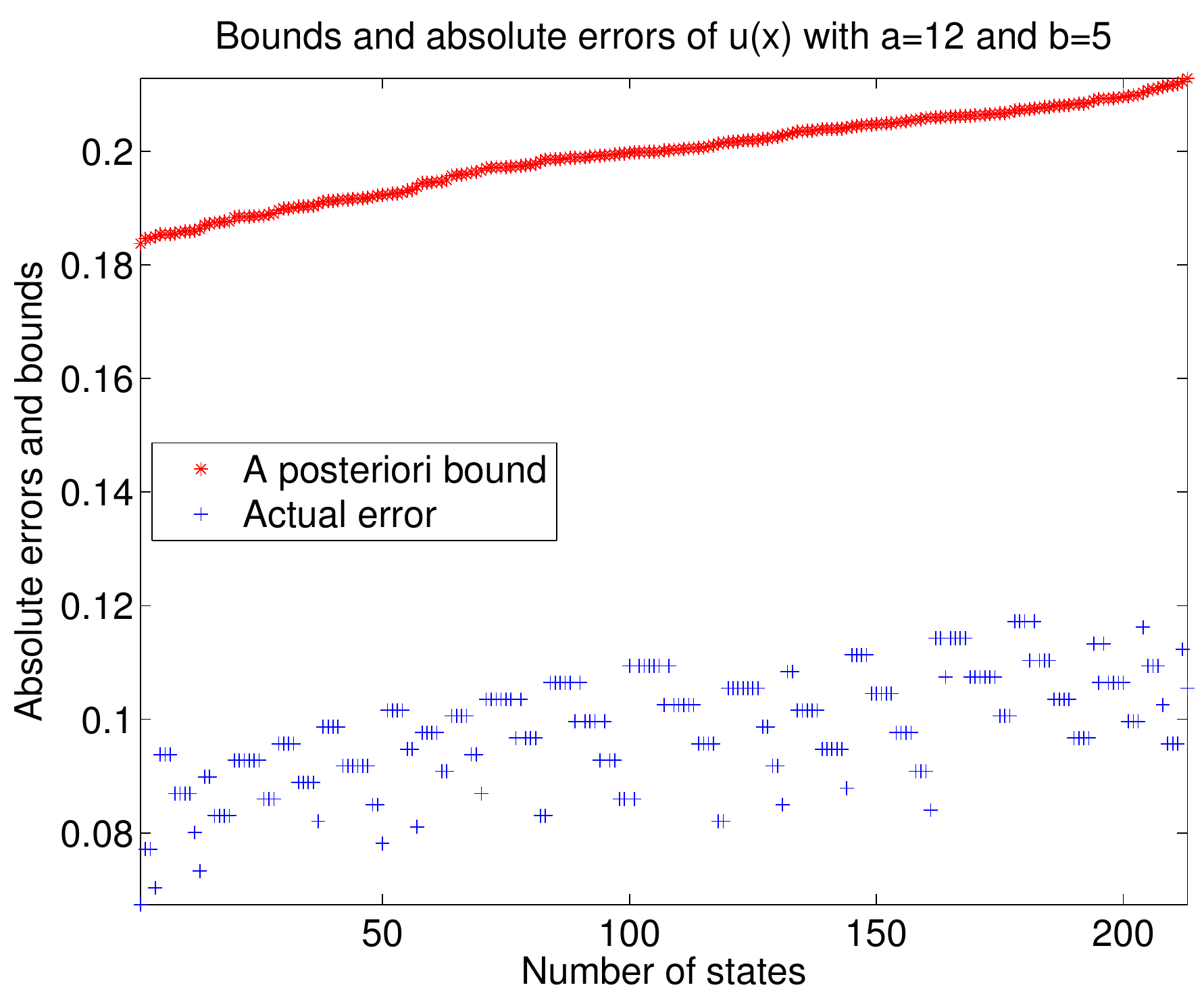} 
 \end{center}
 \caption{A posteriori bounds and real errors with $a=12$ and $b=5$ for the states near the common hyperplane between regions
7 and 13.}\label{fig:5}
 \end{figure}
 		\begin{figure}[ht]
 \begin{center}
 \includegraphics[width=1\linewidth]{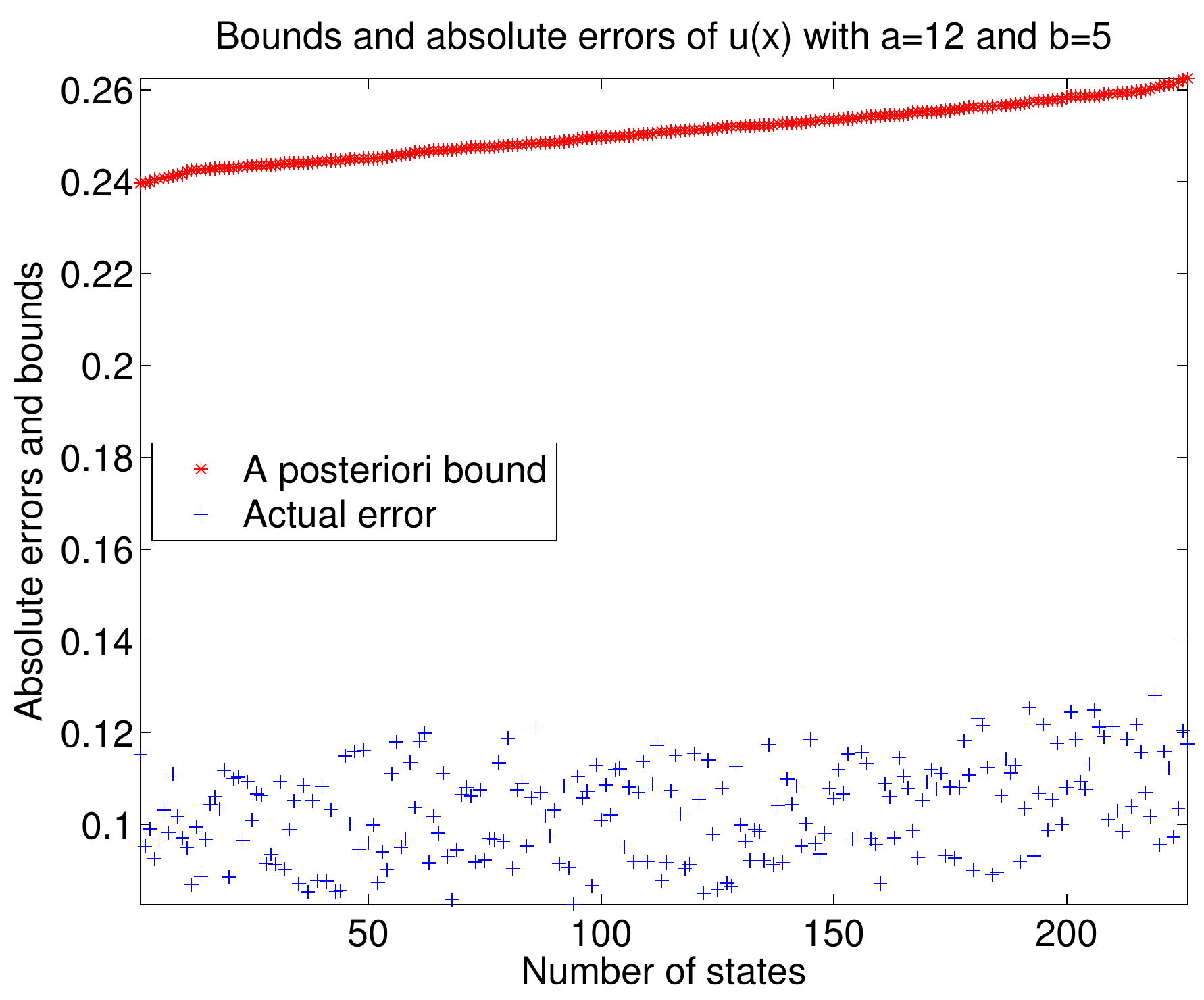} 
 \end{center}
 \caption{A posteriori bounds and real errors with $a=12$ and $b=5$ for the states near the common hyperplane between regions
10 and 13.}\label{fig:6}
 \end{figure}
 
Another nontrivial example is the $P_3$ and $P_{12}$ pair. The maximal a posteriori bound is $0.25$ and the largest actual error is $0.{12}$.  Fig.~\ref{fig:4} plots the a posteriori bounds and the actual errors for the $P_3$ and $P_{12}$ pair.
Fig.~\ref{fig:5} plots the similar data for the $P_7$ and $P_{13}$ pair, and Fig.~\ref{fig:6} for the $P_{10}$ and $P_{13}$ pair. We observe in these figures that our a posteriori bounds overestimate the actual errors about two times.  

We finally note that our bounds in all tests depend on the norm of the state, not on the state itself. Bounding above the norm of the state would make the bounds state-independent.

\section*{Conclusions}\label{sec:conclusion}
EMPC data are quantized and stored in the memory of the controller. 
A state-of-the-art method for determining the precision of the operations in the controller uses a subjective decision based on an ad hoc educated guess of an engineer designing the controller for the given system. Such a subjective decision does not usually guarantee a specific level of accuracy of the control. The validity of the quantized data, representing the system in the controller is typically checked numerically versus the true data on randomly selected state vectors. The state space in true precision includes so many vectors that an exhaustive validation is impractical even off-line. Inaccurate on-line computation of the control can result in suboptimal control, system malfunctions, and failures. 
 
   The actual EMPC control errors can be much larger compared to the quantization error, if the state  is near a hyperplane separating neighboring regions with large and different gains, requiring a special analysis. 
We propose restricting random state vectors to the neighborhoods of the region facets, dramatically decreasing the off-line computational time for numerically checking the control accuracy versus the true data. 
We analyze the EMPC control accuracy deriving upper bounds under an assumption on the quantization error,  typical for a fixed point arithmetic, commonly used in controllers to improve on-line performance. An influence of a rescaling of the state space on the accuracy of the control computation is examined. 
Using our bounds, one can determine the required precision of the quantization and estimate the accuracy in the control input, designing the controller. It is shown that various EMPC data have different sensitivity with respect to the  precision of the quantization, e.g., the gains and offsets are preferred to be stored with high precision. 
We discover that it may be advantageous to use different precisions quantizing various hyperplanes, where the data larger affecting the accuracy of the controller are stored with a greater precision, compared to the data less affecting the accuracy of the controller. Numerical tests for a simple double integrator system support our theory and conclusions. 
Our future work concerns probabilistic approaches, taking into account that the state vector may follow trajectories rather than filling the whole state space.






\bibliographystyle{IEEEtran}
\bibliography{analysis1}

\end{document}